\documentclass[%
 reprint,
 amsmath,amssymb,
 aps,article
]{revtex4-1}

\usepackage{graphicx}
\usepackage{dcolumn}
\usepackage{bm}
\usepackage{physics}
\usepackage{amsthm}
\usepackage[utf8]{inputenc}
\usepackage[english]{babel}

\newtheorem{theorem}{Theorem}

\begin{document}

\preprint{APS/123-QED}

\title{Measurement-device-independent verification of channel steering}

\author{InU Jeon}
\author{Hyunseok Jeong}
\affiliation{%
Department of Physics and Astronomy, Seoul National University, Seoul 08826, Korea}

\date{\today}

\begin{abstract}
  Extending the concept of steerability for quantum states, channel steerability is an ability to remotely control the given channel from a coherently extended party. Verification of channel steering can be understood as certifying coherence of the channel in an one-sided-device-independent manner with respect to a bystander. Here we propose a method to verify channel steering in a measurement-device-independent way. To do this, we first obtain Choi matrices from given channels and use canonical method of measurement-device-independent verification of quantum steering. As a consequence, exploiting channel-state duality which interconverts steerability of channels and that of states, channel steering is verified. We further analyze the effect of imperfect preparation of entangled states used in the verification protocol, and find that threshold of the undesired noise that we can tolerate is bounded from below by steering robustness.
\end{abstract}

\pacs{Valid PACS appear here}
\maketitle


\section{Introduction}
  Quantum steering is a non-classical phenomenon in that local measurements on one side can induce ensembles of local quantum states on the other side which cannot be explained by any classical correlations and local hidden states. This phenomenon has attracted great attention since it was implied in Einstein-Podolsky-Rosen's seminal paper~\cite{Einstein35}, and mentioned by Schr\"odinger~\cite{Schrodinger35} as a \textquoteleft paradox'. Recently, steering has been put into mathematical rigor~\cite{Wiseman07, Jones07}, and a number of theoretical developments~\cite{Cavalcanti09, Branciard12, Chen13, Schneeloch13, He13, Reid13, Wang14, Piani15prl, Uola15, Skrzypczyk15, Nagy16, Kiukas17, Saunders10, Bennet12, Supic16, Gheorghiu17, Ramanathan18} and experimental realizations~\cite{Saunders10, Bennet12, Ramanathan18, Smith12, Wittmann12, Handchen12, Sun14, Armstring15, Sun16, Bartkiewicz16, Tischler18, Cavailles18} have followed.

  The concept of this \textquoteleft \textquoteleft state steering" can be extended to quantum channels~\cite{Piani15}. In the \textquoteleft \textquoteleft channel steering", local measurements on a bystander's side can induce an instrument of reduced channels on the other side which cannot be explained by any classical correlations and an instrument of reduced channels. An operational definition of the steering is verification of non-separability when the steered party is trusted while steering party is untrusted~\cite{Jones07}. Thus an analogous operational definition can be made for the channel steering: the channel steering is a verification of a coherent extension when input and output parties of the reduced channels are trusted while a bystander is untrusted. Constructing the concept of channel steering analogous to that of state steering enables channel-state duality which preserves steerability~\cite{Piani15}. In Ref.~\cite{Piani15}, it was proved that a given quantum channel is steerable if and only if the corresponding quantum state obtained by channel-state duality is steerable. This result links the state steering and the channel steering; the two concepts are not independent topics, rather, they are interconvertible. To clearly show this correspondence, we added some pictures and explanations in Figure 1. One can see the analogy between the concepts of the state steering and channel steering.

  \begin{figure}[ht!]
      \centering
      \includegraphics[width=0.5\textwidth]{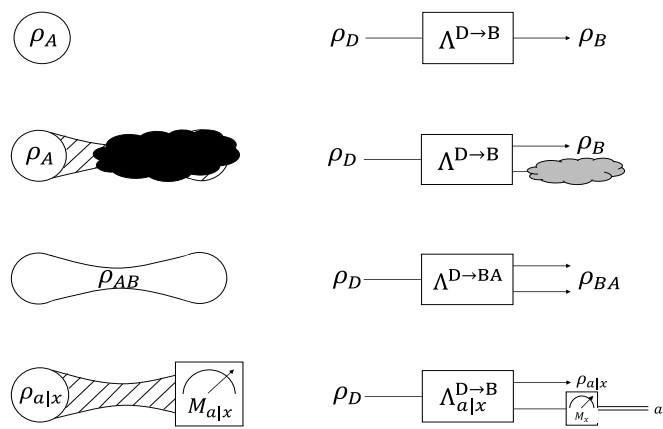}
      \caption{Comparison between state steering and channel steering. (Left) Given a local state $\rho_A$, an extended state $\rho_{AB}$ is such that one can retrieve the original local state $\rho_A$ by discarding a system $B$. If a set of measurements $\{M_{a|x}\}_{a,x}$ are performed on a system $B$, a system $A$ collapses to an assemblage $\{\rho_{a|x}\}_{a,x}$. By checking that whether measurements on a system $B$ really steered a system $A$, we determine steerability of the state. (Right) Channel steering can be defined in an analogous way. Given a channel from system $D$ to $B$, $\Lambda^{D\rightarrow B}$, one can come up with an extended channel with one input $D$ and two output systems $B$ and $A$ such that the original channel $\Lambda^{D\rightarrow B}$ is retrieved by discarding an added output system $A$, a bystander. If a set of measurements $\{M_{a|x}\}_{a,x}$ are performed on a bystander's system, channel reduces to an instrument from $D$ to $B$, $\{\Lambda_{a|x}^{D\rightarrow B}\}_{a,x}$. By checking that whether measurements on a bystander's system really steered a channel from $D$ to $B$, we determine steerability of the channel.}
      \label{figure1}
  \end{figure}

  Meanwhile, in the definitions of steering, whether the parties are trusted or untrusted plays an important role. We mean
by \textquoteleft \textquoteleft Trust"
  that one can obtain complete knowledge of the parties. In the case of state steering, this implies that one can reconstruct the local state of the steered party, while it means for the channel steering that one can reconstruct the reduced channel. In a practical situation, these are unrealizable prerequisites because reliability of experimenters, accuracy of measurement apparatus, and protecting quantum proccessing from external noise are imperfect. Therefore we are only independent of untrusted party, which is called one-sided device-independence(1s-DI). Recently, to alleviate these prerequisites, an alternative verification schemes that remove \textquoteleft \textquoteleft Trust" on the parties in entanglement verification~\cite{Buscemi12} and steering verification~\cite{Cavalcanti13} were proposed. Although the alternative scheme still requires trust on generating device or measurement device of external parties, this does not require perfect control over the experiment of the parties. Thus we can say that it is an alleviated scenario, which is usually called measurement-device-independent(MDI). In the entanglement and steering verifications, canonical techniques to convert entanglement and steering criteria from 1s-DI scenario to the MDI scenario have been well established~\cite{Branciard13, Kocsis15, Ku18}. The MDI method was exploited in quantum key distribution tasks to overcome detector-side channel attacks and obtain better secure distances~\cite{Braunstein12, Lo12}.

  In this paper, motivated by the previous developments, we propose an MDI verification scheme of channel steering. To do this, we first convert the steerability of the channels to that of states with the aid of the channel-state duality. Subsequently, we transform the steering witness in the 1s-DI scenario to the steering criterion in the MDI scenario. We then verify the steerability of the state obtained by channel-state duality using an MDI steering criterion that leads a verification of the steerability of the channel in an MDI manner. Furthermore, we show that the channel-state duality of steerability is still preserved even if we use any bipartite pure state with a full Schmidt-rank, and analyze the effect of imperfect preparation of the state for the channel-state duality. As a result, we find that determination of the channel steerability is still possible for a mixture of the entangled state with any unsteerable noise, and allowed proportion of the noise for a successful verification is bounded by steering robustness introduced in Ref.~\cite{Piani15prl}.

  We begin  with introducing preliminaries for an MDI channel steering verification. In Sec. II, state steering and an MDI verification are explained, and in Sec. III, the concept of channel steering is introduced. To combine the aforementioned concepts, we will introduce channel-state duality in Sec. IV, and show that channel-state duality of steerability is still preserved using any bipartite pure state with full Schmidt-rank. In Sec. V, our scheme to verify channel steering in an MDI way will be proposed and proved. In that section we also consider the effect of imperfect preparation of pure state to obtain dual state of the given channel, and conclude that an MDI verification of channel steering is still possible for some amount of noise, and observe that allowed portion of noise is given by steering robustness. In Sec. VI, we end up with conclusions and further questions for next research. In the paper, we attempt to approach the steering with assemblage and instruments, rather than state and channels, in order to exploit the resource theory of steering in Ref.~\cite{Gallego15}.

\section{quantum state steering}
  In this section, we will introduce some preliminaries to the quantum state steering. A positive semi-definite operator $\rho$ with unit trace is called a quantum state, or a state. A positive semi-definite operator $\rho_a$ with trace less than or equal to unity is called a quantum substate, or a substate. An ensemble $\mathcal{E}$ of a state $\rho$ is a set of substates $\{\rho_a\}_a$ such that the sum of all elements is the state, $\sum_a \rho_a = \rho$. A collection of ensembles $\{\mathcal{E}_x\}_x = \{\rho_{a|x}\}_{a,x}$ is called a state assemblage, or sometimes called an assemblage in short. A state assemblage $\{\rho_{a|x}\}_{a,x}$ is called unsteerable if every substate $\rho_{a|x}$ arise from classical proccessing of some ensemble $\{\rho_\lambda\}_\lambda$ as
    \begin{align}\label{US state}
        \rho_{a|x} = \sum_\lambda p(a|x,\lambda) \rho_\lambda,
    \end{align}
  with some probability distribution $p(a|x,\lambda)$. Normalized substates $\{\rho_\lambda / \Tr[\rho_\lambda]\}_\lambda$ are conventionally referred to as hidden states. If an assemblage is not unsteerable, it is called steerable.

  One can construct similar concepts to observables. A positive-operator-valued-measurement (POVM) is a set of positive semi-definite operators $\{M_a\}_a$ such that the sum of all elements equals an identity, $\sum_a M_a = I$. A collection of POVMs, $\{M_{a|x}\}_{a,x}$, is called a measurement assemblage. A measurement assemblage $\{M_{a|x}\}_{a,x}$ is called jointly measurable or compatible if every POVM element $M_{a|x}$ arise from classical proccessing of some POVM $\{M_\lambda\}_\lambda$, usually named grand POVM, as
    \begin{align} \label{JM}
       M_{a|x}= \sum_\lambda p(a|x,\lambda)\, M_\lambda,
    \end{align}
  with some probability distribution $p(a|x,\lambda)$. If a measurement assemblage is not compatible, it is called incompatible.

It is straightforward to see one-to-one correspondence between state assemblage and measurement assemblage; One can obtain a state assemblage from an extended state by performing measurements in a measurement assemblage on a partial state as
  \begin{align}
      \rho_{a|x}= \mathrm{Tr}_B [\rho_{AB} (I \otimes M_{a|x})],
  \end{align}
  and conversely, one can construct a  measurement assemblage from a state assemblage as
    \begin{align}
        M_{a|x} := (\Tilde{\rho})^{-\frac{1}{2}} \Tilde{\rho}_{a|x} (\Tilde{\rho})^{-\frac{1}{2}},
    \end{align}
  where $\rho = \sum_a \rho_{a|x}$, and tilde denotes an operator projected on $range(\rho)$. Furthermore, it is proved that they have steerability correspondence, that is, a state assemblage is unsteerable if and only if a measurement assemblage in Eq. (4) is compatible~\cite{Uola15}.

  Convex combination of two unsteerable assemblages yields an unsteerable assemblage~\cite{Gallego15}, which means that a set of unsteerable assemblages is convex. Therefore, for every steerable assemblage $\{\rho_{a|x}\}_{a,x}$, we can draw a hyperplane which separates $\{\rho_{a|x}\}_{a,x}$ and a set of unsteerable assemblages. Such a separation is realized by a set of positive semi-definite operators $\{F_{a|x}\}_{a,x}$ called steering witness such that
    \begin{align}\label{SW}
       \sum_{a,x}\Tr[F_{a|x} \rho_{a|x}] \, > \sup_{\{\sigma_{a|x}^{US}\}}\sum_{a,x}\Tr[F_{a|x} \sigma_{a|x}^{US}],
    \end{align}
  where the supremum is taken over a set of unsteerable assemblages. The right hand side of Eq.~(\ref{SW}) is usually called steering bound and conventionally denoted by $\alpha$.

  In an experiment, to obtain knowledge of a state assemblage, one needs to have perfect measurement apparatus because state tomography via imperfect measurements can yield wrong state assemblages different from the real one. Therefore, not to mislead the determination of steerability, we require high reliability of the experimenter. In the case of steering, the steering party is untrusted while the steered party is trusted, thus only the steering party is free from the reliability requirement. This property is called one-sided device-independence (1s-DI) of steering verification. Recently, based on Buscemi's semi-quantum nonlocal games~\cite{Buscemi12}, another scheme of steering verification is introduced~\cite{Cavalcanti13} - an experimenter who has a state assemblage is questioned in quantum states and answers in real numbers. According to experimenter's responses, one can determine whether the state assemblage is steerable or not. This does not require reliability of the experimenter (that is, steered party), thus it is called a measurement-device-independent(MDI) scenario. An MDI steering verification is experimentally demonstrated~\cite{Kocsis15} and canonical way to convert 1s-DI steering witness to that in the MDI scenario is proposed~\cite{Ku18} in analogous way with entanglement witness~\cite{Branciard13}. We will revisit this topic in Sec. V.

\section{Quantum Channel Steering}
  In this section, we will introduce preliminaries to quantum channel steering. A completely positive and trace-preserving linear map $\Lambda[\,\cdot\,]$ is called a quantum channel. Every quantum channel $\Lambda[\,\cdot\,]$ can be written as $\sum_i K_i \,[\, \cdot \,]\, K_i ^\dagger$ with $\sum_i K_i^\dagger K_i = I$, where the operators $\{K_i\}_i$ are called Kraus operators~\cite{Wilde13}. A completely positive and trace non-increasing linear map $\Lambda_a[\, \cdot \,]$ is called a quantum subchannel. Every quantum subchannel $\Lambda_a [\, \cdot \,]$ can be written as $\sum_i K_i \,[\, \cdot \,]\, K_i ^\dagger$ with $\sum_i K_i^\dagger K_i \leq I$, where the operators $\{K_i\}_i$ form a subset of a set of Kraus operators. An instrument of a quantum channel $\Lambda$ is a set of quantum subchannels $\{\Lambda_a \}_a$ such that the sum of every elements forms the quantum channel, $\sum_a \Lambda_a = \Lambda$. A collection of instruments $\{ \Lambda_{a|x} \}_{a,x}$ of quantum channels is called a channel assemblage. A channel assemblage is said to be unsteerable if every subchannel $\Lambda_{a|x}$ arise from classical processing of some instrument $\{ \Lambda_\lambda \}_\lambda$ as
    \begin{align} \label{US channel}
      \Lambda_{a|x}= \sum_\lambda p(a|x,\lambda) \Lambda_\lambda,
    \end{align}
  with some probability distribution $p(a|x,\lambda)$. If a channel assemblage is not unsteerable, it is called steerable.

  For a quantum channel from a system $D$ to $B$, $\Lambda^{D\rightarrow B}$, we can come up with a channel extension with one sender $D$ but two receivers $B$ and $A$, $\Lambda^{D\rightarrow BA}$, which satisfies $\Lambda^{D\rightarrow B}=\mathrm{Tr}_A \circ \Lambda^{D\rightarrow BA}$. If an extended quantum channel $\Lambda ^{D\rightarrow BA}$ can be expressed as a sum of decompositions into quantum state and quantum subchannel,
    \begin{align}
      \Lambda ^{D\rightarrow BA} = \sum_\lambda \Lambda_\lambda^{D\rightarrow B} \otimes \rho_\lambda ^A ,
    \end{align}
  we call it an incoherent extension. If a channel extension is not an incoherent, it is called coherent.

There is a correspondence  between an instrument and a POVM; For every POVM, we can obtain an instrument by performing measurements on one side of a channel extension, and for every instrument, we can find a POVM which induces a given instrument from an extended channel. Along this correspondence, one can consider every channel assemblage $\{\Lambda_{a|x}^{D\rightarrow B}\}_{a,x}$ as an induced one from an extended channel being performed by POVM $\{M_{a|x}^A\}_{a,x}$,
    \begin{align} \label{instrument and measurement}
        \Lambda_{a|x}^{D\rightarrow B}[\,\cdot\,] = \Tr_A[(I \otimes M_{a|x}) \, \Lambda^{D\rightarrow BA}[ \, \cdot \, ]\,],
    \end{align}
  and one can define channel steerability of an extented channel $\Lambda^{D\rightarrow BA}$ as steerability of an instrument of the reduced subchannels $\{\Lambda_{a|x}^{D\rightarrow B}\}_{a,x}$ obtained by Eq.~(\ref{instrument and measurement}).
 As a consequence, the following two statements are equivalent; one is to say that  a channel assemblage is steerable  and the other is  that   a bystander, say Alice, has ability to remotely control the channel $\Lambda^{D\rightarrow B}$ by performing measurements $\{M_{a|x}\}_{a,x}$ on her side.
    Furthermore, it was proved that a verification of the channel steering is equivalent to verifying that a channel extension is coherent when Alice is not a trusted bystander~\cite{Piani15}. This is an analogous argument with the state steering in that steering verification is equivalent to verifying that a shared state is entangled when Alice is not a trusted party~\cite{Jones07}.

\section{Channel-State Duality}

  Aforementioned two concepts of steering can be linked via channel-state duality, or Choi-Jamio\l kowski isomorphism~\cite{Jamiolkowski72,Choi75}. For a quantum channel $\Lambda^{D \rightarrow B} : \mathcal{B}(\mathcal{H}_D) \rightarrow \mathcal{B}(\mathcal{H}_B)$, we can construct an extended quantum state $\rho_{CB} \in \mathcal{B}(\mathcal{H}_C \otimes \mathcal{H}_B)$, where $\mathcal{H}_C$ is a Hilbert space isomorphic to $\mathcal{H}_D$, by
    \begin{align}\label{Choi Matrix}
        \rho^{CB} = (I^C \otimes \Lambda^{D\rightarrow B}) [\,\ket{\Phi}\bra{\Phi}\,].
    \end{align}
  Here, $\ket{\Phi} \in \mathcal{H}_C \otimes \mathcal{H}_D$ is a maximally entangled state
    \begin{align}
        \ket{\Phi} = \frac{1}{\sqrt{d}} \sum_i \ket{ii},
    \end{align}
  where $\{\ket{i}\}_i$ is an orthonormal basis of Hilbert space $\mathcal{H}_D$ (and thus of $\mathcal{H}_C$). In this paper, $\ket{\Phi}$ will be exclusively used to denote a maximally entangled state. The mapping from a quantum channel to a quantum state in the Eq. (\ref{Choi Matrix}) is bijection if we restrict the range of the map to the quantum states with left reduced state is maximally mixed state. A resultant quantum state in Eq. (\ref{Choi Matrix}) is called a Choi matrix of a channel $\Lambda^{D \rightarrow B}$.

  Channel-state duality carries number of interesting correspondences, such as unital maps - states with right reduced state is maximally mixed state, entanglement breaking maps - separable states, and entanglement binding maps~\cite{Horodecki00} - bound entangled states~\cite{Horodecki98}. The correspondence is also shown to be hold for steerability~\cite{Piani15} : A channel extension $\Lambda^{D\rightarrow BA}$ of the channel $\Lambda^{D\rightarrow B}$ is unsteerable if and only if its Choi matrix is unsteerable by any local measurements on $A$. One can re-express this correspondence between channel assemblages and state assemblages : A channel assemblage $\{\Lambda_{a|x}\}_{a|x}$ is unsteerable if and only if a state assemblage consists of their Choi matrices is unsteerable.

  In this paper, we found that steerability correspondence can be obtained, not only by maximally entangled state, but also by any bipartite pure state $\ket{\psi}$ with full Schmidt-rank.  From now on, for simplicity, we will denote the state assemblage obtained from Choi-Jamio\l kowski isomorphism of a channel assemblage using $\ket{\psi}$, $\{(I^C \otimes \Lambda_{a|x}^{D\rightarrow B})[\,\ket{\psi}\bra{\psi}\,]\}_{a,x}$, as $\{\rho_{a|x}^{\ket{\psi}} \}_{a,x}$ and call it a \textit{Choi assemblage obtained by $\ket{\psi}$}.

    \begin{theorem}
       Let $\ket{\psi} \in \mathcal{H}_C \, \otimes \, \mathcal{H}_D$ be a pure state with full Schmidt-rank. Then a channel assemblage $\{\Lambda_{a|x}^{D\rightarrow B}\}_{a,x}$ is unsteerable if and only if a Choi assemblage obtained by $\ket{\psi}$, $\{\rho_{a|x}^{\ket{\psi}} \}_{a,x}$, is unsteerable.
    \end{theorem}

    \begin{proof}
       To prove the sufficiency, let us assume that the given channel assemblage $\{\Lambda_{a|x}^{D\rightarrow B}\}_{a,x}$ is unsteerable. Then we can write $\Lambda_{a|x}^{D\rightarrow B} = \sum_\lambda p(a|x,\lambda)\, \Lambda_\lambda$, and the Choi assemblage obtained by $\ket{\psi}$ reads $\{\sum_\lambda p(a|x,\lambda) \,\, (I \otimes \Lambda_\lambda)[\,\ket{\psi}\bra{\psi}\,]\}_{a,x}$, where $\{\,(\,I \,\otimes \, \Lambda_\lambda \,)[\,\ket{\psi}\bra{\psi}\,]\}_\lambda$ is an ensemble. This is exactly the form of Eq. (1), which proves the sufficiency.

       To prove the necessity, let us assume that a Choi assemblage of the given channel assemblage $\{\Lambda_{a|x}^{D\rightarrow B}\}_{a,x}$ obtained by $\ket{\psi}$ is unsteerable. Then there exists some conditional probability distribution $p(a|x,\lambda)$ and a set of quantum states $\{\sigma_\lambda^{CB}\}_\lambda$ which satisfy
         \begin{align}\label{necessity proof}
            (I^C \otimes \Lambda_{a|x}^{D\rightarrow B})[\,\ket{\psi}\bra{\psi}\,] = \sum_\lambda p(a|x,\lambda) \sigma_\lambda ^{CB}.
         \end{align}
      Polar decomposition of $\sigma_\lambda ^{CB}$ can be written as
        \begin{align}\label{Polar Decomposition}
           \sigma_\lambda ^{CB} = \sum_i p_{i,\lambda} \ket{\psi_{i,\lambda}}\bra{\psi_{i,\lambda}},
        \end{align}
      with some pure states $\ket{\psi_{i,\lambda}} \in \mathcal{H}_C \otimes \mathcal{H}_B$ and probability distribution $p_{i,\lambda}$. Using the property of maximally entangled state, let us express pure states as
        \begin{align}\label{pure state expression}
            \ket{\psi} &= (I \otimes K) \ket{\Phi} =(K ^T \otimes I) \ket{\Phi},\\
            \ket{\psi_{i,\lambda}} &= (I \otimes K_{i,\lambda}) \ket{\Phi} =(K_{i,\lambda} ^T \otimes I) \ket{\Phi},
        \end{align}
      with some operators $K$ and $K_{i,\lambda}$ in $B(\mathcal{H}_D)$ (and thus in $B(\mathcal{H}_C)$). From Eq. (\ref{necessity proof}) and (13), summing over $a$ and tracing out a system $B$ yields $(K^\dagger K)^T / d = \sum_\lambda \sigma_\lambda^C$. Meanwhile, from Eq. (\ref{Polar Decomposition}) and Eq. (14), tracing out a system $B$ gives us $\sigma_\lambda^C = \sum_i p_{i,\lambda} (K_{i,\lambda} ^\dagger K_{i,\lambda})^T / d$. By combining the two results, we obtain that $(K^\dagger K)^T = \sum_{i,\lambda} p_{i,\lambda} (K_{i,\lambda} ^\dagger K_{i,\lambda})^T$ which guarantees that $\{\sqrt{p_{i,\lambda}}K_{i,\lambda} K^{-1}\}_{i,\lambda} $ is a set of Kraus operators. Therefore the quantum states $\sigma_\lambda^{CB}$ can be written as
        \begin{align}
            \sigma_\lambda^{CB} = (I^C \otimes \Lambda_\lambda^{D\rightarrow B}) [\,\ket{\psi}\bra{\psi}\,],
        \end{align}
      where $\Lambda_\lambda^{D\rightarrow B} [ \, \cdot \, ] = \sum_i p_{i,\lambda} K_{i,\lambda} K^{-1}\, [\,\cdot \,]\, (K^{-1})^\dagger K_{i,\lambda}^\dagger$ is a subchannel and  $\{\Lambda_\lambda\}_\lambda$ is an instrument. Furthermore, it can be easily shown that a mapping
        \begin{align} \label{bijectivity of filter}
           A \mapsto K^{-1} \, A \, (K^{-1})^\dagger
        \end{align}
      is bijective. Therefore, Eq. (\ref{necessity proof}), Eq. (15), bijectivity of the map (\ref{bijectivity of filter}) and the channel-state duality complete the necessity proof.

    \end{proof}
   We note here that the steerability correspondence proved in Ref.~\cite{Piani15} can be derived as a special case of the Theorem 1, obtained by setting $K = I$.

\section{MDI verification of channel steering}
In order  to determine the steerability of a channel assemblage $\{\Lambda_{a|x}^{D\rightarrow B}\}_{a,x}$, one needs to perform channel tomography to ascertain whether it can be obtained from classical processing of some instrument, as in the form of Eq. (\ref{US channel}). However, channel tomography requires perfect control over input and output systems. For example, experimenters must be reliable, their generating device for input states and measurement device for tomography should be accurate enough. Consequently, for deficient apparatus or untrustworthy experimenters, one cannot be assured of results of steerability determination. Therefore, to alleviate the demands, an MDI verification of channel steering is needed, thus in this section, we will propose how to verify channel steering in an MDI manner.

  \begin{figure}[hb!]
    \centering
    \includegraphics[width=0.5\textwidth]{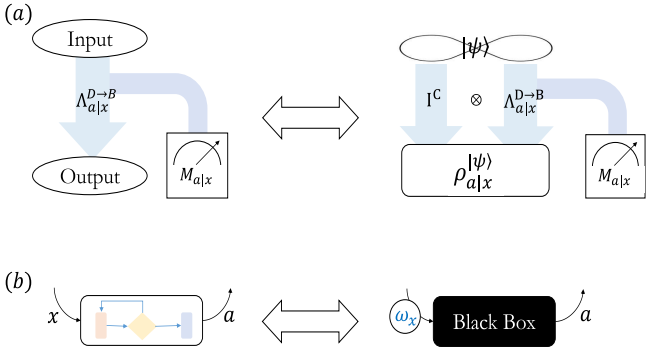}
    \caption{Schematic of an MDI verification protocol of channel steering. (a) Due to Theorem 1, steerability of the channel assemblage (Left) is equivalent to the steerability of the Choi assemblage obtained by $\ket{\psi}$ (Right). Thus we first obtain Choi assemblage using bipartite pure states $\ket{\psi}$ with full Schmidt-rank, and verify steerability of the resultant Choi assemblage $\{\rho_{a|x}^{\ket{\psi}}\}_{a,x}$. (b) (Left) 1s-DI verification of state steering can be accomplished by asking and receiving classical information from the experimenter when the whole process can be trusted. It is depicted as a white box which takes and yields classical information. (Right) If we encode the question in quantum states, MDI verification of state steering is possible without any trust on the process. It is depicted as a black box because we do not need to care about what is taking place inside. Therefore, by verifying obtained Choi assemblage in the MDI protocol, the MDI verification of channel steering is accomplished.}
    \label{figure2}
\end{figure}

  For the first step, we will verify steerability of a Choi assemblage obtained by $\ket{\psi}$ instead of that of the channel assemblage, by virtue of the steerability correspondence proved in the preceding section. Suppose that a channel assemblage $\{\Lambda_{a|x}^{D\rightarrow B}\}_{a,x}$ of which one wants to determine steerability is given. One can obtain the corresponding Choi assemblage obtained by $\ket{\psi}$, $\{\rho_{a|x}^{\ket{\psi}} \}_{a,x}$, by preparing bipartite pure states with full Schmidt-rank $\ket{\psi}\bra{\psi} \in B(\mathcal{H}_C \otimes \mathcal{H}_D)$ and transmitting the $D$ part through the subchannels $\Lambda_{a|x}^{D\rightarrow B}$, while preserving the $C$ part. In this process, one who supervises the experiment actively participates in the verification protocol by generating quantum states, thus he/she will be named as a \textit{referee} from now on. To verify steerability of the Choi assemblage obtained by $\ket{\psi}$ in an MDI way, we will adopt canonical method of an MDI verification protocol proposed in Refs.~\cite{Branciard13, Ku18}. Let $\{F_{a|x}\}_{a.x}\subset B(\mathcal{H}_C \otimes \mathcal{H}_B)$ be a steering witness for a Choi assemblage $\{\rho_{a|x}^{\ket{\psi}} \}_{a,x}$ such that
    \begin{align}
        \sum_{a,x} \Tr[F_{a|x} \, \rho_{a|x}^{\ket{\psi}}] > \sup_{\{\sigma_{a|x}^{US}\}} \sum_{a,x} \Tr[F_{a|x} \, \sigma_{a|x}^{US}] := \alpha,
    \end{align}
  where the supremum is taken over all unsteerable assemblages $\{\sigma_{a|x}^{US}\}_{a,x} \subset B(\mathcal{H}_C \otimes \mathcal{H}_B) $. The existence of such a witness is guaranteed by the convexity of unsteerable assemblages~\cite{Gallego15} and hyperplane separation theorem. We can decompose a steering witness into
    \begin{align}\label{SW Decomposition}
        F_{a|x} - \frac{\alpha}{|\mathcal{X}|}I = \sum_{z,y} \beta^{z,y,x}_{1,1,a} \,\, \tau_z  \otimes \omega_y
    \end{align}
  for any $a$ and $x$, where $\mathcal{X}$ is an index set of $x$, $\{\tau_z\}_z \subset B(\mathcal{H}_C)$ and $\{\omega_y\}_y \subset B(\mathcal{H}_B)$ are tomographically complete sets of states. Here, decomposition of the steering witness needs not be unique. To accomplish the MDI verification of steerability, the referee prepares quantum states $\tau_z^T \in B(\mathcal{H}_{C'})$ and $\omega_y^T\in B(\mathcal{H}_{B'})$, where $\mathcal{H}_{X'}$ is a Hilbert space isomorphic to $\mathcal{H}_{X}$ for $X \in \{C,B\}$, and provides them to parties $C$ and $B$, say Charlie and Bob, respectively. Charlie (Bob) thereupon tries to maximize a value
    \begin{align} \nonumber \label{Score}
       &I(\{\rho_{a|x}^{\ket{\psi}} \}_{a,x},\beta) \\= &\sum_{a,x,y,z} \beta^{z,y,x}_{1,1,a} \,\, \Tr [ (Q^{C'C} \otimes P^{BB'})(\tau_z^T \otimes \rho_{a|x}^{\ket{\psi}} \otimes \omega_y^T)],
    \end{align}
  which is usually called score, by performing joint measurements $Q^{C'C}$ ($P^{BB'}$) on the given state $\tau_z^T$ ($\omega_y^T$) and his part of the output substate $\rho_{a|x}^{\ket{\psi}}$. As a consequence, it can be shown that the score (\ref{Score}) is positive for a steerable assemblage $\{\rho_{a|x}^{\ket{\psi}} \}_{a,x}$ while non-positive for any unsteerable assemblage $\{\sigma_{a|x}^{US}\}_{a,x}$. A sketch of the MDI verification protocol of the channel steering is provided in Figure 2. In this figure, we used white-box and black-box pictures to focus on the \textquoteleft trust' on the process and the type of questions, while simplifying other details.

  \textit{proof}

  Let us first prove that (\ref{Score}) is non-positive for unsteerable assemblages. For any unsteerable assemblage $\{\sigma_{a|x}^{US}\}_{a,x}$, one can write $\sigma_{a|x}^{US} = \sum_\lambda p(a|x,\lambda) \sigma_\lambda^{CB}$ for some probability distribution $p(a|x,\lambda)$ and ensemble $\{\sigma_\lambda^{CB}\}_\lambda$, as in the Eq. (\ref{US state}). Therefore, for any joint POVM $Q^{C'C}$, $P^{BB'}$, the score (\ref{Score}) reads
    \begin{align}
        &I(\{\sigma_{a|x}^{US}\}_{a,x},\beta) = \sum_{a,x,y,z,\lambda} \beta^{z,y,x}_{1,1,a} p(a|x,\lambda)  \nonumber
        \\  \cross & \,\, \Tr[(Q^{C'C} \otimes P^{BB'})(\tau_z^T \otimes \sigma_\lambda^{CB} \otimes \omega_y^T)] \nonumber
        \\ = & \sum_{a,x,y,z,\lambda} \beta^{z,y,x}_{1,1,a} \,\, p(a|x,\lambda) \,\, \Tr [ (R_\lambda^{C'B'})(\tau_z^{C'} \otimes \omega_y^{B'})^T],
    \end{align}
  where $R_\lambda^{C'B'}$ is a reduced POVM element, defined by $R_\lambda^{C'B'} = \Tr_{CB}[\,(Q^{C'C} \otimes P^{BB'})(I \otimes \sigma_\lambda^{CB} \otimes I) \,]$. Since a reduced POVM element is a positive semi-definite operator, one can convert the reduced POVM element to a substate $\rho_{R,\lambda} := (R_\lambda^{C'B'})^T / N$ so that $\{ \rho_{R,\lambda}\}_\lambda$ forms an ensemble, where $N$ is a normalizing constant, $N = \Tr[Q^{C'C} \otimes P^{BB'}]$. Substituting $R_\lambda^{C'B'}$ by $\rho_{R,\lambda}$, the score reads

    \begin{align}  \nonumber
         &I(\{\sigma_{a|x}^{US}\}_{a,x},\beta)\\ \nonumber
         =& \Tr[ \sum_{a,x,\lambda} \, p(a|x,\lambda) \,\, (R_\lambda^{C'B'})^T \,\, \sum_{z,y} \, \, \beta^{z,y,x}_{1,1,a} \,\, (\tau_z^{C'} \otimes \omega_y^{B'})] \\ \nonumber
         =& N \,\,\sum_{a,x} \Tr[\sum_{\lambda} \, p(a|x,\lambda) \,\, \rho_{R,\lambda} \,\, (F_{a|x} - \frac{\alpha}{\mathcal{|X|}}I)] \\
         =& N \,\, \sum_{a,x} \Tr[\rho_{a|x}^{US} \,\, (F_{a|x} - \frac{\alpha}{\mathcal{|X|}}I) ] \,\, \leq 0,
    \end{align}
  which proves the latter part of the statement, that is, the score is non-positive for any unsteerable assemblage. Here, $\rho_{a|x}^{US} \, = \, \sum_\lambda p(a|x,\lambda) \rho_{R,\lambda}$, thus the last inequality in the Eq. (21) holds by definition of the steering witness.

  Meanwhile, for the Choi assemblage $\{\rho_{a|x}^{\ket{\psi}} \}_{a,x}$, Charlie and Bob can obtain positive score by projecting the system $C'C$ and $BB'$ into maximally entangled state $\ket{\Phi}\bra{\Phi}$,

  \begin{align}
      &I(\{\rho_{a|x}^{\ket{\psi}} \}_{a,x},\beta) \nonumber
      \\ = & \sum_{a,x,y,z} \beta^{z,y,x}_{1,1,a} \, \Tr[\,(\ket{\Phi}\bra{\Phi} \otimes \ket{\Phi}\bra{\Phi})(\tau_z^T \otimes \rho_{a|x}^{\ket{\psi}} \otimes \omega_y^T)\,] \nonumber
        \\ = & \sum_{a,x} \,\, \Tr [ \sum_{z,y} \,\, \beta^{z,y,x}_{1,1,a} \,\, (\tau_z \otimes \omega_y) \,\, \rho_{a|x}^{\ket{\psi}} ]\,\, / \, {d_C d_B} \nonumber
        \\ = & \sum_{a,x} \Tr[(F_{a|x} - \frac{\alpha}{\mathcal{|X|}}I) \,\, \rho_{a|x}^{\ket{\psi}}] \,\, / \, {d_C d_B} \, > 0 ,
  \end{align}
  where we used used the property of maximally entangled state, $\Tr[\, \ket{\Phi}\bra{\Phi} (A \otimes B)\,] = \Tr[A^TB]\,\,/\, d_A$, for the second equality. The last inequality is derived by the definition of the steering witness. Therefore we can always obtain positive score for the Choi assemblage $\{\rho_{a|x}^{\ket{\psi}} \}_{a,x}$, which completes the proof that the score (\ref{Score}) is a steering criterion in the MDI scenario. $\square$

  Here, in order to avoid misleading, we point out that the MDI verification of channel steering is \textit{not} free from trust or perfect control over the experiment. Although MDI scenario is independent from measurement device, still it requires generation of quantum states - bipartite pure states with full Schmidt-rank $\ket{\psi}\bra{\psi}$ to obtain the Choi assemblage from the channel $\Lambda^{D\rightarrow B}$, and sets of tomographically complete quantum states $\{\tau_z\}_z$ and $\{\omega_y\}_y$. Thus we still need trust on or control over the generating devices. However, we can still say that the MDI verification is better than the original scenario for following reasons. First, in the 1s-DI scenario, it is hard to test reliability assumptions such as trustworthy experimenter, accurate measurement apparatus, perfect control over the experiment. On the contrary, in the MDI scenario, the assumption - generation of quantum states - is open to a test from an external party. Any ombudsman can bring their own measurement apparatus and test generated quantum states via state tomography. In this sense, we can say that the MDI verification is more reliable. Second, in the MDI scenario, measurement inefficiency of the experiment does not affect steerability at all, along similar line with the MDI entanglement verification~\cite{Branciard13}. When losses occur, the only effect to the score is that the probability part, $Tr [ (Q^{C'C} \otimes P^{BB'})(\tau_z^T \otimes \rho_{a|x}^{CB} \otimes \omega_y^T)]$, is multiplied by measurement efficiencies of Charlie and Bob. Since positive factor does not change the sign, score is still non-positive for any unsteerable assemblage and positive for the Choi assemblage $\{\rho_{a|x}^{\ket{\psi}} \}_{a,x}$. This guarantees loss tolerance of the MDI steering verification with respect to Charlie and Bob. Third, in the MDI scenario, imperfect generation of quantum states can be analyzed and its effect on verification protocol can be quantified. For imperfect preparation of quantum question, when type of the questions and form of the score are set in two-qubit system, the effect to the score is exactly quantified in Ref.~\cite{Kocsis15} and generalized for inefficient measurements in Ref.~\cite{Jeon19}. Furthermore, for imperfect generation of pure states with full Schmidt-rank which are used for channel-state duality, although we cannot obtain perfect Choi assemblage obtained by $\ket{\psi}$, we can still verify channel steering in MDI way as follows. Let us suppose that we failed to perfectly generate full Schmidt-rank pure state, instead, let us say that we generated mixtures of full Schmidt-rank pure state and colored noise with relative weight $0 \leq w\leq 1$,
    \begin{align}
        \rho_w^{CD} = w \ket{\psi}\bra{\psi} + (1-w)\sigma_{colored},
    \end{align}
  where $\sigma_{colored} = \sum_i p_i \ket{ii}\bra{ii}$ is a colored noise. This is the case frequently occur in laboratories since colored noise is a result of decoherence of the bipartite pure state $\ket{\psi}\bra{\psi} = \sum_{i,j} \sqrt{p_ip_j} \ket{ii}\bra{jj}$. One can observe that if all off-diagonal parts of the pure state $\ket{\psi}\bra{\psi}$ is washed out, it ends up in the colored noise.

  If we use Eq. (23) for obtaining channel-state duality, that is, if we transmit $D$ part to the subchannel and preserve $C$ part, we obtain
    \begin{align} \label{Choi assemblage with colored noise}
        &(I^C \otimes \Lambda_{a|x}^{D\rightarrow B}) (\rho_w^{CD}) \nonumber
        \\=\, w \, \rho_{a|x}^{\ket{\psi}} + &(1-w)\, (I \otimes \Lambda^{D\rightarrow B}_{a|x}) \left[ \, \sigma_{colored} \, \right],
    \end{align}
  where $\{ \rho_{a|x}^{\ket{\psi}} \}_{a,x}$ is the Choi assemblage obtained by $\ket{\psi}$ from the given channel assemblage  $\{\Lambda_{a|x}^{D\rightarrow B}\}_{a,x}$. One can consider Eq. (\ref{Choi assemblage with colored noise}) as a convex combination of two state assemblages $\{ \rho_{a|x}^{\ket{\psi}} \}_{a,x}$ and $\{(I \otimes \Lambda^{D\rightarrow B}_{a|x}) \left[\,  \sigma_{colored} \, \right] \}_{a,x}$. We observe that the colored noise is separable, thus unsteerable. By the result of the theorems in Ref.~\cite{Gallego15}, local operations do not increase the steerability of a state assemblage, thus $\{(I \otimes \Lambda^{D\rightarrow B}_{a|x}) \left[ \, \sigma_{colored} \, \right] \}_{a,x}$ remains unsteerable. Meanwhile, from the definition of steering measures such as steerable weight~\cite{Skrzypczyk14} or steering robustness~\cite{Gallego15}, steering measure is zero if and only if the assemblage is unsteerable. Furthermore, forenamed steering measures are convex monotones~\cite{Gallego15}, which means that for any real number $0 \leq r \leq 1$, and two assemblages $\{\rho_{a|x}\}_{a,x}$, $\{\sigma_{a|x}\}_{a,x}$,
  \begin{align}
      S(\,r \rho_{a|x} + (1-r)\, \sigma_{a|x}) \leq \, r \, S(\rho_{a|x}) + (1-r)\, S(\sigma_{a|x}),
  \end{align}
  where $S$ is the steering measure. As a consequence, nonzero steering measure of Eq. (\ref{Choi assemblage with colored noise}) implies nonzero steering measure of the Choi assemblage obtained by $\ket{\psi}$, $\{ \rho_{a|x}^{\ket{\psi}} \}_{a,x}$. Therefore, by verifying steerability of Eq. (\ref{Choi assemblage with colored noise}), we end up with verifying channel assemblage $\{\Lambda^{D\rightarrow B}_{a|x}\}_{a,x}$.

  The preceding argument does not only cover colored noise, but also deals with all type of noise which is unsteerable, because what we exploited is the property that local operations do not increase steerability. Let us assume that the generated state is in the form of
    \begin{align}
      \rho_w^{CD} = w \ket{\psi}\bra{\psi} + (1-w)\sigma^{US},
    \end{align}
  where $\sigma^{US}$ is an unsteerable noise. Then after transmitting $D$ part to the subchannel and preserving $C$ part, we obtain
    \begin{align} \label{Choi assemblage with unsteerable noise}
        &(I^C \otimes \Lambda_{a|x}^{D\rightarrow B}) [\rho_w^{CD}] \nonumber
        \\=\, w \, \rho_{a|x}^{\ket{\psi}} + &(1-w)\, (I \otimes \Lambda^{D\rightarrow B}_{a|x}) \left[ \, \sigma^{US} \, \right].
    \end{align}
  Following the same argument as before, we conclude that verification of steerability of the Eq. (27) implies that of the channel assemblage.


How much unsteerable noise can we then  tolerate? In other words, at least how much noise is required to obliterate the steerability of the assemblage? The lower bound of answer of this question is determined by \textit{steering robustness of the assemblage} $R(\{\rho_{a|x}\})$,
    \begin{align} \label{SR}
        R(&\{\rho_{a|x}\}) = \inf_{\{\sigma_{a|x}\}} \, t \,\,\,\, s.t\nonumber
        \\ \frac{\rho_{a|x}+ t\, \sigma_{a|x}}{1+t}& \,\, \text{is an unsteerable assemblage,}
    \end{align}
  where infimum is taken over a set of all assemblages. If we restrict $\{\sigma_{a|x}\}$ to a set of unsteerable assemblages, the obtained value, say $R^{US}(\{\rho_{a|x}\})$, will be larger than or equal to $R(\{\rho_{a|x}\})$, and determine the amount of allowed noise by $\frac{R^{US}(\{\rho_{a|x}^{\ket{\psi}}\})}{1+R^{US}(\{\rho_{a|x}^{\ket{\psi}}\})}$. Therefore, the threshold of the noise in generating pure states $\ket{\psi}$ to verify steerability of the channel assemblage is at least $\frac{R(\{\rho_{a|x}^{\ket{\psi}}\})}{1+R(\{\rho_{a|x}^{\ket{\psi}}\})}$. This tells us that, even if any bipartite pure state with full Schmidt-rank can be used to obtain channel-state duality, in practical situation we should use states of which their Choi assemblage has high steering robustness to endure undesired noise.

  At this point, we can raise three natural questions. First, how much would be a gap between $R$ and $R^{US}$? Second, which pure state $\ket{\psi}$ yields the largest steering robustness (and $R^{US}$) for the given channel assemblage $\{\Lambda^{D\rightarrow B}_{a|x}\}_{a,x}$? (A mixed state cannot be such a state, due to the linearity of quantum channel.)  We conjecture that it would be a maximally entangled state $\ket{\Phi}$, independent of a given channel, while we leave its proof for future work. Third, does a Choi assemblage preserve order of steering measure of pure states? That is, for a given channel assemblage $\{\Lambda^{D\rightarrow B}_{a|x}\}_{a,x}$ and for pure states $\ket{\psi}$ and $\ket{\phi}$ such that $S(\ket{\psi}) \leq S(\ket{\phi})$, where $S(\ket{X})$ is a steering measure of pure state $\ket{X}$, would it be true that $S(\{\rho_{a|x}^{\ket{\psi}}\}) \leq S(\{\rho_{a|x}^{\ket{\phi}}\})$? These are interesting open questions.

\section{Conclusion}

  We proposed a way to verify channel steering in an MDI manner. We first converted a channel assemblage to a state assemblage via Choi-Jamio\l kowski isomorphism using bipartite pure state with full Schmidt-rank, and then determined steerability of the state assemblage. Afterwards, we applied canonical method of MDI verification of state steering to Choi assemblage obtained by $\ket{\psi}$, and showed that the steering criterion, called score, is non-positive for any unsteerable state assemblages while it is positive for the Choi assemblage. As a consequence of steerability correspondence we proved, this verifies steerability of the given channel assemblage.

  We further analyzed the situation of imperfect preparation of a pure state with full Schmidt-rank for obtaining Choi assemblage. A common case would be that colored noise is generated and thus a pure state is contaminated, just like a Werner state. We showed that not only for colored noise, but also for any type of unsteerable noise, can we verify channel steering in an MDI manner for a portion of undesired noise bounded from below by steering robustness of a state assemblage.

  Our work leads to a couple of natural questions as follows. (1) How much is a gap between $R$ and $R^{US}$? (2) Which pure state gives the maximal steering robustness for a given channel assemblage? (3) Does a Choi assemblage preserve the order of the steering measure for pure states? We leave these open problems for future research.


\end{document}